\documentclass[10pt, twocolumn]{IEEEtran}

\usepackage{epsfig,latexsym}
\usepackage{amsmath}
\usepackage{amssymb}
\usepackage{subfigure}
\usepackage[colorlinks]{hyperref}
\usepackage{indentfirst}
\usepackage{times}
\usepackage{fancyhdr}
\usepackage{lastpage}
\usepackage{bigints}
\usepackage{subfigure}
\usepackage{graphicx}
\usepackage{amsthm}

\newtheorem{lemma}{Lemma}
\newtheorem{remark}{Remark}

\usepackage[noend]{algpseudocode}
\usepackage{algorithmicx,algorithm}
\usepackage{color}
\usepackage{ulem}

\begin{document}
\title{Codebook-based Port Selection and Combining for CSI-Free Uplink Fluid Antenna Multiple Access}
\author{Chenguang Rao, 
            Kai-Kit Wong,~\IEEEmembership{Fellow,~IEEE}, 
            Sai Xu,~\IEEEmembership{Member,~IEEE},\\
            Xusheng Zhu,
            Yangyang Zhang, and 
            Chan-Byoung Chae,~\IEEEmembership{Fellow,~IEEE}
\vspace{-8mm}

\thanks{The work of C. Rao, K. K. Wong, S. Xu and X. Zhu was supported by the Engineering and Physical Sciences Research Council (EPSRC) under grant EP/W026813/1.} 
\thanks{The work of C.-B. Chae was in part supported by the Institute for Information and Communication Technology Planning and Evaluation (IITP)/National Research Foundation of Korea (NRF) grant funded by the Ministry of Science and Information and Communications Technology (MSIT), South Korea, under Grant RS-2024-00428780 and 2022R1A5A1027646.}

\thanks{C. Rao, K. K. Wong, S. Xu and X. Zhu are with the Department of Electronic and Electrical Engineering, University College London, London, United Kingdom. K. K. Wong is also affiliated with the Yonsei Frontier Lab, Yonsei University, Seoul 03722 South Korea (e-mail: $\rm \{chenguang.rao,kai\text{-}kit.wong,sai.xu,xusheng.zhu\}@ucl.ac.uk$).}
\thanks{Y. Zhang is with Kuang-Chi Science Limited, Hong Kong SAR, China (e-mail: $\rm yangyang.zhang@kuang\text{-}chi.org$).}
\thanks{C.-B. Chae is with the School of Integrated Technology, Yonsei University, Seoul 03722, South Korea (e-mail: cbchae@yonsei.ac.kr).}

\thanks{Corresponding author: Kai-Kit Wong.}
}
\maketitle

\begin{abstract}
Fluid antenna multiple access (FAMA) has recently emerged as a simple, promising scheme for large-scale multiuser connectivity, offering strong scalability with low implementation complexity. Nevertheless, most existing FAMA studies focus on downlink transmission under perfect channel state information (CSI) at the receiver side, while the uplink counterpart remains largely unexplored. This paper proposes a novel codebook-based port selection and combining (CPSC) FAMA framework for the uplink communications without CSI at the base station (BS). In the proposed scheme, a predefined codebook is designed and broadcast by the BS. Each user equipment (UE) employs a fluid antenna, acquires its local CSI and independently chooses the most suitable codeword, activates the corresponding fluid antenna ports, and determines the combining weights to achieve a two-way match between the selected codeword and the instantaneous effective channel. The BS then separates the superimposed user signals through codebook-guided projection operations without requiring global CSI or multiuser joint optimization. To handle potential codeword collisions, three lightweight scheduling strategies are introduced, offering flexible trade-offs between signaling overhead and collision avoidance. Simulation results demonstrate that the proposed CPSC-FAMA approach achieves substantially higher rates than fixed-antenna systems while maintaining low complexity. Moreover, the results confirm that amortizing the optimization cost over the UEs effectively reduces the BS processing burden and enhances scalability, making the proposed scheme a strong candidate for future sixth-generation (6G) networks.
\end{abstract}


\vspace{-2mm}
\section{Introduction}
\subsection{Context}
\IEEEPARstart{E}{volving} beyond fifth-generation (B5G) and into sixth-generation (6G) wireless technologies, the three famous use cases namely enhanced mobile broadband (eMBB), ultra-reliable low-latency communication (URLLC), and massive machine-type communication (mMTC) are becoming the new six use scenarios \cite{6G1,6G2,6G4}, as specified by the International Mobile Telecommunications-2030 framework (IMT-2030) \cite{ITUwhite}. Out of the six use scenarios, massive connectivity represents a classical wireless communication problem that has bothered the industry for decades. Massive connectivity seeks to share the same channel use over many user equipments (UEs), which is greatly motivated by the push for higher spectral efficiency and the massive internet-of-things (IoT) scenarios \cite{Shafique-2020,Guo-2021}.\footnote{In this paper, the terms `UE' and `IoT device' mean the same thing and will be used interchangeably depending on the context.}

State-of-the-art relies on multiuser multiple-input multiple-output (MIMO) precoding in the downlink \cite{Villalonga-2022} and interference rejection combining (IRC) in the uplink \cite{Wassie-2015}. While having a large antenna array at the base station (BS) is a default setting and will continue to be so in next-generation cellular networks \cite{Wang-xlmimo,MMIMO2}, serious issues regarding their scalability cannot be neglected. Hardware complexity aside, one issue is the need for accurate channel state information (CSI) at the BS to work properly. This will become unbearable if the number of UEs is extremely large and even more so if there is no infrastructure for massive IoT communication. In the latter, the IoT devices may not even be synchronized, let alone processed jointly with their CSI at the BS or gateway when they report their data to the BS/gateway. Recent directions on holographic and dynamic metasurface antennas can reduce the power consumption for large-scale MIMO systems but do little in reducing the signal processing burden for multiple access \cite{HMA,DMA,tri1,tri2}.

The past few years have also witnessed increasing interest in non-orthogonal multiple access (NOMA) \cite{Saito-2013,Dai-2015} and rate-splitting multiple access (RSMA) \cite{RSMA}, which many advocate to be a massive connectivity solution. However, both NOMA and RSMA would require perfect CSI at both ends to operate reliably because of the use of interference cancellation. Also, in practice, heavy channel coding would be required to ensure correct operation of the interference cancellation, minimizing error propagation, but this would compromise the capacity.

\vspace{-3mm}
\subsection{Scalable Multiple Access}
Evidently, there needs to be a rethnking of massive multiple access---one that should be light in CSI and signal processing requirements at the BS for scalability. One emerging scheme toward that direction is fluid antenna multiple access (FAMA) \cite{Rabie-2024,Wong-2022}. FAMA exploits antenna position reconfigurability to mitigate interference entirely without the need of precoding nor interference cancellation. FAMA is an application of the fluid antenna system (FAS) concept to multiple access.

FAS represents the paradigm that treats antenna as a reconfigurable physical-layer resource, broadening system design and optimization \cite{FAS1,FAS2,FAS3,FAS4,Lu-2025,New-2026jsac,Wong-2026wc}. FAS is hardware agnostic and can appear in numerous forms, such as liquid-based antennas \cite{shen2024design,Shamim-2025}, metamaterials \cite{Liu-2025arxiv,Zhang-jsac2026}, pixel-based structures \cite{zhang2024pixel,tong-2025pixel} and others. Since the seminal work in \cite{wong2021fluid,wong2020pel}, many efforts have been made to understand the diversity benefits of FAS, see e.g., \cite{Khammassi2023,espinosa2024anew,New2023fluid,new2023information}. CSI estimation in FAS has also been investigated in recent attempts using a range of techniques \cite{xu2023channel,new2025channel,zhang2025successive}.

One of the most important innovations in FAS is arguably FAMA. The position reconfigurability of FAS can provide UE the ability to scan through the ups and downs of the received signal with fine resolution in the spatial domain, and because of that, the UE can take advantage of the spatial opportunity where the interference signal suffers from a deep fade due to multipath fading, naturally without relying on precoding nor interference cancellation. The benefit is that the burden of the BS for CSI acquisition and complex precoding optimization is completely eliminated in the downlink. There have been good progress in developing this FAMA idea. For instance, \cite{Wong-2022} and \cite{FAMA3} proposed the fast FAMA technique that switches the FAS port on a per-symbol basis to receive the signal where the sum-interference cancels by itself. Later, \cite{FAMA1,FAMA2} devised the slow FAMA scheme which tunes the FAS port by maximizing the average signal-to-interference plus noise ratio (SINR) and hence is much more practical. Learning-based methods were also proposed to realize slow FAMA with reduced CSI \cite{FAMAN1}. In addition, slow FAMA has also been shown to be effective in channel-coded multi-carrier systems \cite{EFAMA1}. Further advances in slow FAMA were made in \cite{CUMA,PCUMA,FAMAN3} by choosing many ports and combining their signals in the analog domain, illustrating a considerable gain in multiplexing. The work in \cite{EFAMA2} further investigated FAMA systems in a cell-free setting. Most recently, \cite{EFAMA4} addressed the optimal design of antenna configuration and power control in FAMA networks while \cite{FAMAN2} applied FAMA in time-varying ionospheric channels.

\vspace{-6mm}
\subsection{Research Gap}
Despite the steadily growing body of work on FAMA, most results to date focus on the downlink, with only a handful of studies addressing the uplink case. For instance, \cite{UFAMA1} focused on minimizing the total transmit power in an uplink FAS by optimizing the BS antenna positions. Later, \cite{UFAMA2} considered an uplink multiple-access channel in which each UE is equipped with a fluid antenna array and the BS employs fixed-position antennas, and the paper approximated the system's channel capacity by jointly optimizing the data rates for all users. After that, \cite{UFAMA3} analyzed a FAMA-aided wireless-powered system, where UEs employ fluid antennas for both energy harvesting and uplink transmission. In \cite{UFAMA4}, FAS was applied to assist a NOMA uplink where interference cancellation was imperfect. A deep reinforcement-learning-based optimization scheme was proposed to jointly adjust the BS beamforming, users' power allocation, and antenna positions, achieving improved sum-rate performance under residual interference. Finally, in \cite{UFAMA5}, a two-timescale uplink MIMO-FAS transmission framework was presented under imperfect CSI, where antenna positions were optimized using long-term statistical CSI while beamforming vectors were adapted to instantaneous CSI. 

Overall, it seems customary to consider the availability of CSI at the BS, which despite a standard assumption, does little in improving scalability of massive multiple access for future wireless networks. In particular, we are interested in:
\begin{quote}
\begin{center}
{\em Would it be possible to handle massive multiple access in the uplink without CSI at the BS?}
\end{center}
\end{quote}

This is greatly motivated by the situations where reliable CSI acquisition becomes infeasible, e.g., in infrastructure-less massive IoT applications, or even in cellular settings where CSI is fast changing and too many UEs are involved.

Without CSI at the BS, IRC such as minimum mean square error (MMSE), zero-forcing (ZF) and advanced receivers such as maximum-likelihood detector (MLD) or interference cancellation, are no longer possible. If there are methods to ensure that the BS resolves the inter-user interference with a massive number of UEs without CSI, this could fundamentally address the scalability problem for massive connectivity.

\vspace{-3mm}
\subsection{Our Objective}
This paper proposes a novel uplink transmission framework, named codebook-based port selection and combining (CPSC) FAMA, which utilizes a codebook-based communication and decoding scheme to enable CSI-free operation at the BS. The BS is equipped with multiple fixed antennas while each UE has a FAS. In the proposed scheme, each UE is required to have its local CSI, and {\em independently} selects a codeword, activates the corresponding ports, and determines its combining weights based on its own channel condition. At the BS, the superimposed received signal can be efficiently separated into individual UE components through simple projection operations, without the need for CSI estimation. Furthermore, unlike existing FAMA approaches that depend on joint optimization among users, all optimization in the proposed framework is performed locally at each UE.\footnote{In the downlink, FAMA schemes typically do not require joint optimization among UEs \cite{Wong-2022,FAMA1,EFAMA1,CUMA,PCUMA,FAMAN3} but in the uplink, when multiple fluid antennas are deployed at the BS, this will require joint optimization of port selection and signal combining over all the UEs at the BS.} Each UE requires only a single RF chain, which also means significant hardware savings. This decentralized structure effectively distributes computational and signaling loads from the BS to the UEs, enabling higher scalability as the number of UEs increases. As a consequence, the proposed approach exhibits strong potential to support large-scale multiuser access in future wireless networks. Our main contributions are summarized as follows: 
\begin{itemize}
\item A novel CPSC-FAMA framework is proposed, in which a predefined codebook is configured and broadcast by the BS, from which each UE selects its optimal codeword according to its local CSI. With the reconfigurability of FAS, each UE activates appropriate ports and determines the weighting vector to achieve a two-way match between the selected codeword and its instantaneous CSI. A low-complexity and efficient algorithm is designed to realize this local optimization process. At the BS, the received composite signal is processed using the predefined codebook, enabling the separation of multiple UEs' messages through simple projection operations without requiring CSI. Furthermore, a simplified version without UE-side combining is also investigated, where corresponding algorithms are developed to reduce hardware cost.
\item To handle the situations when different UEs choose the same codeword, three scheduling strategies are proposed. Specifically, a random deferral scheme allows collided UEs to retransmit in subsequent slots, a user-reselection scheme lets them autonomously choose alternative codewords within the current slot. Besides, a BS-reassignment scheme allocates unoccupied codewords directly. These strategies provide flexible trade-offs between signaling overhead, latency, and collision-avoidance efficiency.
\item Simulation results verify the effectiveness of the proposed framework. The impacts of port selection, codeword optimization, combining, and collision-handling strategies are comprehensively evaluated. The results demonstrate that the proposed CPSC-FAMA scheme achieves superior average rate and scalability compared with conventional fixed-antenna benchmarks, while maintaining low complexity and signaling overhead. Moreover, the simulation results also provide valuable insights into the trade-offs between performance, complexity, and system load, thus confirming the feasibility of the proposed scheme for large-scale 6G random-access scenarios.
\end{itemize}

The remainder of this paper is organized as follows. Section~\ref{sec:model} introduces the system model and the basic framework of CPSC-FAMA. Section \ref{sec:codeword} provides an effective strategy for UEs to choose an appropriate codeword from the codebook. Section~\ref{sec:Sb} proposes a low-complexity algorithm for selecting the optimal ports and the combining vector. Additionally, an algorithm that optimizes only the port selection while keeping the combining vector fixed is developed to keep hardware resource requirements low. Then Section~\ref{sec:collision} introduces three scheduling schemes to address the codeword collision issues. Section \ref{sec:sim} then discusses simulation results, and Section \ref{sec:conclu} concludes the paper.  The main symbols used throughout this paper are summarized as follows: $x$, $\mathbf{x}$, $\mathbf{X}$ represent scalar, vector, and matrix, respectively. $\lVert\mathbf{x}\rVert$ represents the magnitude of $\mathbf{x}$. $(\cdot)^{T}$, $(\cdot)^{H}$ denote the transpose and Hermitian transpose, respectively. $\Re\{\cdot\}$ represents the real part, and $j_{0}(\cdot)$ denotes the zeroth-order Bessel function of the first kind.

\begin{figure}
\centering
\includegraphics[width=1\linewidth]{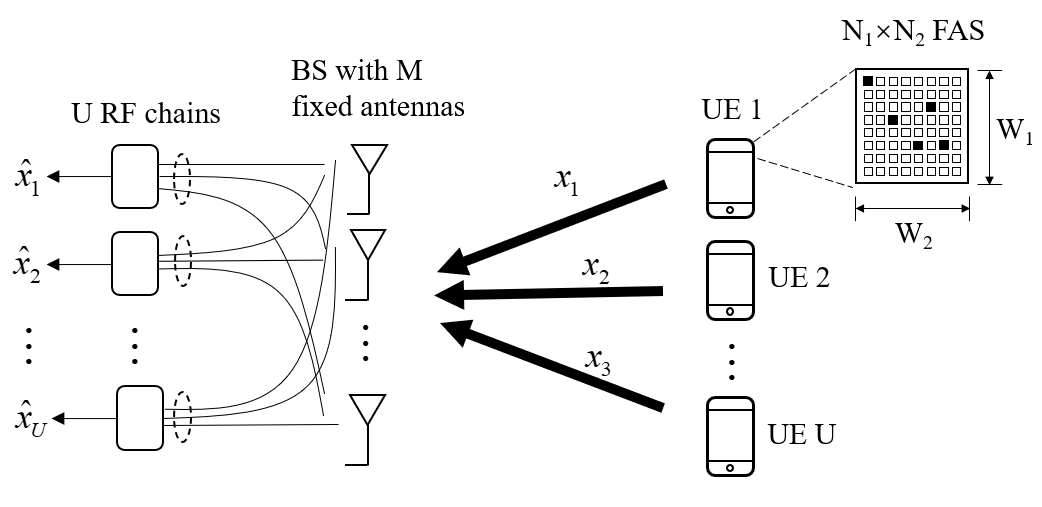}
\vspace{-4mm}
\caption{An uplink communication system, where the BS is equipped with \(M\) fixed-antenna, and the UEs are equipped with fluid antennas.}\label{fig:model}
\vspace{-2mm}
\end{figure}

\vspace{-2mm}

\section{System Model and\\The Codebook-based Framework}\label{sec:model}
\subsection{System Model}
Consider an uplink communication system with an \(M\)-fixed antenna BS and \(U\) UEs, \(M\geq U\), as shown in Fig.~\ref{fig:model}. Each UE is equipped with a 2D fluid antenna array having \(N = N_1\times N_2\) discrete ports. The size of FAS is \(W = W_1\lambda\times W_2\lambda\), where \(\lambda\) represents the carrier wavelength. Let \(\mathbf{H}_u\in\mathbb{C}^{M\times N}\) denote the uplink channel from the \(N\) ports of UE \(u\) to the BS, which can be expressed as
\begin{equation}\label{eq:ch}
\mathbf{H}_u=\big[\mathbf{h}_{u,1},\dots,\mathbf{h}_{u,N}\big]\in\mathbb{C}^{M\times N},
\end{equation}
where \(\mathbf{h}_{u,n}=[h^{(u)}_{1,n} ~h^{(u)}_{2,n}\cdots ~h^{(u)}_{M,n}]^T\in\mathbb{C}^{M}\) represents the channel vector from the \(n\)-th FAS port of UE \(u\) to the BS, with $h^{(u)}_{m,n}$ being the complex channel coefficient between the $m$-th BS antenna and the $n$-th port of UE $u$. As in \cite{New2023fluid,new2023information}, we have used the conversion $(n_1,n_2)\mapsto n=k_{n_1,n_2}$ so that the ports are ordered in one dimension. In particular,
\begin{equation}\label{eq:2d_2_1d}
k_{n_1,n_2} = n_2+(n_1-1)N_2.
\end{equation}

Under a Rician fading model, \(\mathbf{h}_{u,n}\), can be expressed as
\begin{equation}
\mathbf{h}_{u,n} = \sqrt{\frac{L}{1+L}}\bar{\mathbf{h}}_{u,n} + \sqrt{\frac{1}{1+L}}\tilde{\mathbf{h}}_{u,n}.
\end{equation}
in which \(L\) is the Rice factor, {\color{black}\(\bar{\mathbf{h}}_{u,n}\)} and \(\tilde{\mathbf{h}}_{u,n}\) represent the line-of-sight (LoS) path gain and non-LoS (NLoS) path gain, respectively. Under  rich scattering, \([\tilde{\mathbf{h}}_{u,n}]_m\) is a zero-mean complex Gaussian random variable, whose spatial covariance over any two ports is given by \(\mathbb{E}\{h^{(u)}_{m,k_1}\,(h^{(u)}_{m,k_2})^*\}\), so that
\begin{multline}
J_{k_1,k_2} = \\ 
\Omega j_0\left(2\pi\sqrt{\left(\frac{n_1-\widetilde{n}_1}{N_1-1}W_1\right)^2+\left(\frac{n_2-\widetilde{n}_2}{N_2-1}W_2\right)^2}\right),
\end{multline}
where \(\Omega = \mathbb{E}\{|h_{m,u}^{(k)}|^2\}\) measures the channel power, $k_1=k_{n_1,n_2}$ and $k_2=k_{\widetilde{n}_1,\widetilde{n}_2}$ which can be obtained by (\ref{eq:2d_2_1d}).

During the transmitting phase, each UE activates \(K_u\) multiple ports. Define the port activation matrix \(\mathbf{A}_u\in\{0,1\}^{N\times K_u}\) for UE \(u\), which is given by
\begin{equation}
\mathbf{A}_u = \left[\mathbf{a}_{u,1},\mathbf{a}_{u,2},\dots,\mathbf{a}_{u,K_u}\right].
\end{equation}
Each column of \(\mathbf{A}_u\) is a standard basis vector, i.e., \(\mathbf{a}_{u,m}\in\{\mathbf{e}_1,\dots,\mathbf{e}_{N}\}\). Also, for all \(1\leq i\neq j \leq K_u\), we have \(\mathbf{a}_{u,i}\neq\mathbf{a}_{u,j}\). Then the effective uplink channel is \(\boldsymbol{h}_u\triangleq\mathbf{H}_u\mathbf{A}_u\mathbf{b}_u\in\mathbb{C}^{M}\), where \(\mathbf{b}_u\in\mathbb{R}^{K_u}\) is the combining vector. It is assumed that the per-UE combining is power-normalized with \(\|\mathbf{b}_u\|_2=1\).\footnote{The hardware implication of having the combining vector at the UE will be discussed later. But it is worth pointing out that the combining vector is real-valued, meaning that there are no phase shifters involved but simple attentuation of the received signals from the FAS ports. We will also consider the case where $\mathbf{b}_u$ is fixed and such scaling is avoided in Section \ref{subsec:ncb}.} Let \(x_u\) be the transmitted symbol of UE \(u\) with \(\mathbb{E}\{|x_u|^2\}=P_u\), and \(\mathbf{n}\sim\mathcal{CN}(\mathbf{0},\sigma^2\mathbf{I}_M)\) be the additive white Gaussian noise. Each UE sends a message to the BS simultaneously. Then the received signal at BS can be expressed as
\begin{equation}\label{eq:yBS}
\mathbf{y}=\sum_{u=1}^{U}\boldsymbol{h}_u x_u+\mathbf{n}=\mathbf{H}_{\mathrm{eff}}\mathbf{x}+\mathbf{n},
\end{equation}
where \(\mathbf{H}_{\mathrm{eff}}=[\boldsymbol{h}_1,\dots,\boldsymbol{h}_U]\in\mathbb{C}^{M\times U}\) and \(\mathbf{x}=[x_1,\ldots,x_U]^T\). After all the UEs have transmitted, the BS collects the received signals with \(U\) RF chains and extracts the data \(\hat{x}_1\sim\hat{x}_U\) as the estimates of the original messages \(x_1\sim x_U\) separately. 

\begin{figure}
\centering
\includegraphics[width=1\linewidth]{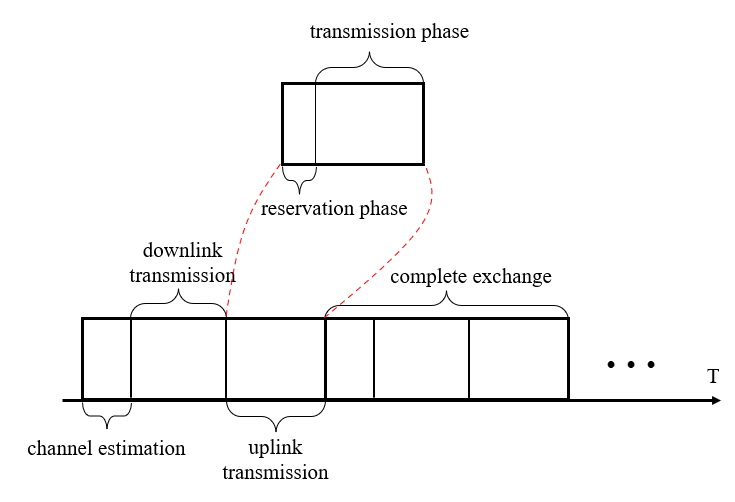}
\vspace{-4mm}
\caption{A complete frame, including a channel estimation slot, a downlink transmission slot, and an uplink transmission slot. The uplink slot consists of a reservation phase and a transmission phase.}\label{fig:time_model}
\vspace{-3mm}
\end{figure}

In this work, we assume that each UE has the knowledge of its own channel \(\mathbf{H}_u\), while the BS has no CSI. Consider a complete exchange that consists of a BS downlink transmission followed by the UE's uplink reply within one coherence block, as illustrated in Fig.~\ref{fig:time_model}. Also, assume the channel does not change within one block, i.e., the channel can be treated as static during this exchange. During the downlink, the UE acquires the downlink channel \(\mathbf H_u^{\mathrm{DL}}\) using techniques such as those in \cite{xu2023channel,new2025channel,zhang2025successive}. Under time-division-duplex (TDD) reciprocity with proper calibration, the uplink channel equals the transpose of the downlink one, i.e., \(\mathbf H_u = \big(\mathbf H_u^{\mathrm{DL}}\big)^T\). In the sequel, our design will be based on this assumption. 

\vspace{-4mm}
\subsection{A Codebook-Based Combining Scheme}
To extract a specific UE's message from the superimposed signal of all the UEs received, the BS needs to do so without CSI. To achieve this, a codebook-based scheme is employed. In particular, the BS first predefines a shared codebook that contains \(M\) orthogonal codewords. This codebook is known to all UEs. During uplink access, each UE independently selects one codeword from the codebook and configures its FAS ports to match the codeword. For the BS to identify which codeword is selected by each UE, each time slot is divided into two phases: a reservation phase and a transmission phase, as shown in Fig.~\ref{fig:time_model}. In the reservation phase, each UE selects a codeword from the predefined codebook and transmits a short reservation signal with the index of the codeword it chooses. At this stage, UEs can transmit messages on orthogonal resources, e.g., in the time domain. Since the signals to be transmitted are very short, this cost is acceptable. After scanning, the BS learns the codeword utilization status and can subsequently process the received signals in the transmission phase, where each UE transmits its data message using the selected codeword. 

We denote \(\mathbf Q=[\mathbf{q}_1,\ldots,\mathbf{q}_M]\in\mathbb C^{M\times M}\) as the codebook known to both sides. UE \(u\) selects the codeword \(\mathbf{q}_{k(u)}\). The BS then projects \(\mathbf{y}\) in (\ref{eq:yBS}) onto \(\mathbf{q}_{k(u)}\), and obtains
\begin{equation}\label{eq:zku}
\begin{aligned}
z_{k(u)} &=\mathbf{q}_{k(u)}^{H}\mathbf y \\ 
&= \underbrace{\mathbf{q}_{k(u)}^{H}\boldsymbol{h}_u x_u}_{\text{desired signal}} +  \underbrace{\sum_{v=1\atop v\neq u}^{U}\mathbf{q}_{k(u)}^{H}\boldsymbol{h}_v x_v}_{\text{interference}}+  \underbrace{\mathbf{q}_{k(u)}^{H}\mathbf{n}}_{\text{noise}}.
\end{aligned} 
\end{equation}
Here, we have assumed that each UE selects a different index. The case where multiple UEs select the same codeword will be discussed in Section \ref{sec:collision}. Ideally, the interference terms are zero, while the signal power remains undiminished. Therefore, when \(\mathbf{q}_{k(u)} = \boldsymbol{h}_u/\|\boldsymbol{h}_u\|_2\) for all \(1\leq u\leq U\), \ref{eq:zku} is written as
\begin{equation}\label{eq:zku2}
z_{k(u)} = |\boldsymbol{h}_u| x_u + \mathbf{q}_{k(u)}^{H}\mathbf{n}.
\end{equation}

Since the BS does not have any CSI, we may choose any unitary codebook (e.g., DFT). The FAS enables UE to adjust the channel \(\boldsymbol{h}_u\) to match a codeword from the codebook by changing the activated ports and setting the combining vector. This will be discussed in the next subsection.

\vspace{-2mm}
\subsection{Problem Formulation}\label{subsec:PF}
Recalling that \(\boldsymbol{h}_u=\mathbf{H}_u\mathbf{A}_u\mathbf{b}_u\), it follows that \(\boldsymbol{h}_u\) lies within the column space of \(\mathbf{H}_u\), \(\operatorname{span}\{\mathbf{H}_u\}\). Apparently, if \(\exists 1\leq m\leq M\), \(\operatorname{rank}\{\mathbf{H}_u\} = \operatorname{rank}\{\big[\mathbf{H}_u, \mathbf{q}_m\big]\}\), then \(\mathbf{q}_m\in\operatorname{span}\{\mathbf{H}_u\}\). In this case, it is optimal for UE \(u\) to select \(\mathbf{q}_m\) as the codeword. In the ideal case where \(\mathrm{rank}(\mathbf{H}_u)=M\), we have \(\mathrm{span}(\mathbf{H}_u)=\mathbb{C}^{M}\). Thus, any target beam \(\mathbf{q}_k\in\mathbb{C}^{M}\) can be matched exactly by activating \(M\) linearly independent ports. Concretely, pick any index set \(S\) with \(|S|=M\) such that \(\mathbf{H}_{u,S} = \mathbf{H}_{u}[:,S] \in \mathbb{C}^{M\times M}\) is full rank, activate those ports according to the index, and compute the combining vector directly by
\begin{equation}\label{eq:bu}
\mathbf{b}_{u,0} = \left(\mathbf{H}_{u,S}^{H}\mathbf{H}_{u,S}\right)^{-1}\mathbf{H}_{u,S}^{H}\mathbf{q}_k,~\mbox{with }
\mathbf{b}_u = \frac{\mathbf{b}_{u,0}}{\|\mathbf{b}_{u,0}\|_2}.
\end{equation}
This yields \(\mathbf{H}_{u}\mathbf{A}_u\mathbf{b}_u=\mathbf{H}_{u,S}\mathbf{b}_u = \alpha\mathbf{q}_k\), i.e., indicating a perfect alignment. However, when the number of users \(U\) is large, to lower the chance of multiple UEs picking the same codeword, one typically needs \(M\gg U\). As activating \(M\) ports per UE would then be burdensome on the UE hardware, this method is not suitable. As a result, in the remainder of this subsection, we consider a constrained design with a fixed and much smaller number of active ports \(K_u\ll M\).

The problem consists of three parts: codeword selection, port selection, and combining vector optimization, i.e.,
\begin{equation}\label{eq:p0}
P_0:\left\{\begin{aligned}
 \min_{\mathbf{b}_u,\mathbf{A}_u,k(u)} &~
\left\| \frac{\mathbf{H}_{u}\mathbf{A}_u\mathbf{b}_u}{\lVert \mathbf{H}_{u}\mathbf{A}_u\mathbf{b}_u\rVert_2}-\mathbf{q}_{k(u)}\right\|_2^2\\
\mbox{s.t.} &~\mathbf{b}_u\in\mathbb{R}^{K_u}, \left\| \mathbf{b}_u\right\|_{2}=1,\\
		&~\mathbf{A}_u = \left[\mathbf{a}_{u,1},\dots,\mathbf{a}_{u,K_u}\right]\in\{0,1\}^{N\times K_u},\\
		&~\mathbf{a}_{u,m}\in\{\mathbf{e}_1,\dots,\mathbf{e}_{N}\},\\
		&~1\leq k(u)\leq M.
\end{aligned}\right.
\end{equation}
The target is to make the normalized effective channel vector \(\boldsymbol{h}_u = \mathbf{H}_{u}\mathbf{A}_u\mathbf{b}_u\) as close as possible to the codeword vector \(\mathbf{q}_{k(u)}\). According to the expansion formula for the magnitude of subtraction of any two complex vectors \(\mathbf{x},\mathbf{y}\in\mathbb{C}^K\):
\begin{equation}
\left\|\mathbf{x}-\mathbf{y} \right\|_2^2 = \left\|\mathbf{x} \right\|_2^2 + \left\|\mathbf{y} \right\|_2^2 - 2\Re\{\mathbf{y}^H\mathbf{x}\},
\end{equation}
the objective function can be transformed to
\begin{equation}
\left\|\frac{\mathbf{H}_{u,S}\mathbf{b}_u}{\left\|\mathbf{H}_{u,S}\mathbf{b}_u\right\|_2}-\mathbf{q}_{k(u)}\right\|_2^2 
= 2-2\Re\left\{\mathbf{q}_{k(u)}^H\frac{\mathbf{H}_{u,S}\mathbf{b}_u}{\left\|\mathbf{H}_{u,S}\mathbf{b}_u\right\|_2}\right\}.
\end{equation}
For conciseness, we define the port selection set
\begin{equation}\label{eq:Sdf}
S = \{1\leq s_1 < s_2 < \dots < s_{K_u} \leq N\}.
\end{equation} 
Then by denoting \(\mathbf{H}_{u,S} = \mathbf{H}_{u}\mathbf{A}_u = \mathbf{H}_{u}[:,S]\), the problem can be reexpressed as 
\begin{equation}\label{eq:p1}
P_1:\left\{\begin{aligned}
\max_{\mathbf{b}_u,S,k(u)} &~
\Re\left\{\mathbf{q}_{k(u)}^H\frac{\mathbf{H}_{u,S}\mathbf{b}_u}{\left\|\mathbf{H}_{u,S}\mathbf{b}_u\right\|_2}\right\}\\
\mbox{s.t.} &~\mathbf{b}_u\in\mathbb{R}^{K_u}, \left\|\mathbf{b}_u\right\|_{2}=1\\
		&~\mathbf{H}_{u,S} = \mathbf{H}_{u}[:,S],\\
		&~S = \{1\leq s_1 < s_2 < \dots < s_{K_u} \leq N\},\\
		&~1\leq k(u)\leq M.
\end{aligned}\right.
\end{equation}
Problem $P_1$ can be better understood by dividing it into three parts with three optimization variables: 
\begin{itemize}
\item Finding the codeword index \(k(u)\).
\item Finding the port set \(S\) with a fixed \(k(u)\).
\item Finding the combining vector \(\mathbf{b}_u\) with fixed \(S\) and \(k(u)\).
\end{itemize}
These three parts will be handled separately next.

\vspace{-2mm}
\section{Selection of the Codeword Index}\label{sec:codeword}
To select the best codeword at UE $u$, a straightforward way is to exhaustively evaluate all candidates from the codebook \(\mathbf{Q}\) and for each codeword, compare the objective function after optimizing \(S\) and \(\mathbf{b}_u\) and pick the best one. However, in massive access scenarios, the number of available codewords will be extremely large, which leads to unacceptable complexity. Therefore, a more effective method needs to be sought.

Following the idea from \cite{PCUMA}, we adopt the approach with matrix projection. As mentioned in Section \ref{subsec:PF}, the ideal case is when one codeword lies in \(\operatorname{span}\{\mathbf{H}_u\}\). If such a codeword does not exist, a suboptimal solution is to select the codeword closest to this \(\operatorname{span}\{\mathbf{H}_u\}\). Intuitively, if a candidate \(\mathbf{q}_{k(u)}\) can be well represented by some \(K_u\)-column subspace of \(\mathbf{H}_u\), then \(\mathbf{q}_{k(u)}\) is likely to already have a significant component within the column space of \(\mathbf{H}_u\). Therefore, a suboptimal rule to select the codeword \(\mathbf{q}_{k(u)}\) can be achieved by
\begin{equation}\label{eq:kus0}
k(u)^{\star} \approx \operatorname*{arg\,max}_{\substack{1\leq k(u)\leq M}}\|\mathbf{P}_{\mathbf{H}_u}\mathbf{q}_{k(u)}\|_2^2,
\end{equation}
where \(\mathbf{P}_{\mathbf{H}_u} = \mathbf{H}_u(\mathbf{H}_u^H\mathbf{H}_u)^{-1}\mathbf{H}_u^H\) represents the orthogonal projector of \(\operatorname{span}\{\mathbf{H}_u\}\). However, using \eqref{eq:kus0} requires about \(O(NM^2+M^2N+N^3)\) operations to compute \(\mathbf{P}_{\mathbf{H}_u}\), which is still unacceptable. To simplify the problem, singular value decomposition (SVD) can be employed to obtain a solution with a low-rank approximation. Actually, the leading singular vectors of \(\mathbf{H}_u\) capture the directions where the columns of \(\mathbf{H}_u\) concentrate most of their energy, and the trailing singular directions are the weak components that cannot be reliably approximated by any small subset of columns. Therefore, by screening the candidate set \(\mathbf{Q}\) by their projection onto the top-\(t\) left singular vectors of \(\mathbf{H}_u\) (with \(t\) slightly larger than \(K_u\)), an approximated solution can be obtained. In particular, let \(\mathbf{H}_u=\mathbf{U}_u\mathbf{\Sigma}_u\mathbf{V}_u^H\) be the SVD of \(\mathbf{H}_u\), and denote by \(\mathbf{U}_{u,t}\in\mathbb{C}^{M\times t}\) the first \(t\) left singular vectors that correspond to the top-\(k\) largest singular values. The codeword \(\mathbf{q}_{k(u)}\) can then be selected according to the following rule:
\begin{equation}\label{eq:kus1}
k(u)^{\star} \approx \operatorname*{arg\,max}_{\substack{1\leq k(u)\leq M}}\|\mathbf{U}_{u,t}^H\mathbf{q}_{k(u)}\|_2^2.
\end{equation}

To obtain the first \(t\) left singular vectors of \(\mathbf{H}_u\), we apply Randomized SVD \cite{RSVD1,RSVD2}, whose complexity are only at \(O(MN\log t + Nt^2)\). Combined with the cost for projection, the total computational cost to obtain \(k(u)^{\star}\) according to \eqref{eq:kus1} is \(O(MN\log t + Nt^2) + O(M^2t)\). At this point, the complexity is mainly caused by the projection operation. By choosing an appropriate codebook, the complexity can be further reduced. In particular, when the codebook \(\mathbf{Q}\) is generated by discrete Fourier transform (DFT), we have
\begin{multline}
\mathbf{q}_m =\\
\left[1,e^{-j\frac{2\pi}{M}(m-1)},e^{-j\frac{2\pi}{M}2(m-1)},\dots, e^{-j\frac{2\pi}{M}(M-1)(m-1)}\right]^T.
\end{multline}
Therefore, the projection operation \(\mathbf{U}_{u,t}^H\mathbf{Q}\) can be regarded as applying an \(M\)-point DFT on each row of \(\mathbf{U}_{u,t}^H\), which can be computed effectively via fast Fourier transform (FFT). The complexity of projection with FFT is only \(O(tM\log M)\), which is much lower than \(O(M^2t)\).

\begin{remark}\label{rem:nq}
The choice of \(t\) can balance performance and complexity. A larger \(t\) would improve the performance while requiring more operations. Comparisons for different choices of \(t\) will be presented in Section \ref{sec:sim}.
\end{remark}

\vspace{-2mm}
\section{Port Selection And Combining Methods}\label{sec:Sb}
Once an appropriate codeword is selected, each UE needs to determine the set of activated ports \(S\) and the combining vector \(\mathbf{b}_u\). This section investigates the selection schemes for \( S \) and \(\mathbf{b}_u\). Specifically, \(\mathbf{b}_u\) can be derived in closed form if \(S\) is fixed. 
Also, we explore the selection strategy of \(S\) in cases where combining is not performed and \(\mathbf{b}_u = \frac{1}{\sqrt{K_u}}\mathbf{1}\).

Before we address the optimization, it is worth mentioning that the fact that \(\mathbf{b}_u\) is real-valued and its power is constrained to be one, means that each branch from a selected port only needs to apply passive attenuation rather than amplification before the signals from all the selected ports are superimposed. Therefore, the whole process can be realized on a single RF chain. Possible implementations include Rotman lenses \cite{Rotman-1963,Gherbi-2024}, time-division multiplexing (TDM) based RF combiners \cite{An-2014}, electrically steerable parasitic array radiator \cite{Luther-2012,Oh-2017}, and broadband variable-gain amplifiers (VGAs) \cite{Sobotta-2015}. 

\subsection{General Solution of \(\mathbf{b}_u\) and \(S\)}\label{subsec:cb}
Once the codeword \(\mathbf{q}_{k(u)}\) and the port set \(S\) are chosen, the problem becomes a single-variable optimization problem:
\begin{equation}\label{eq:pb}
P_{\mathbf{b}_u}: \max_{\mathbf{b}_u\in\mathbb{R}^{K_u}\atop \|\mathbf{b}_u\|_{2}=1}
\Re\left\{\mathbf{q}_{k(u)}^H\frac{\mathbf{H}_{u,S}\mathbf{b}_u}{\lVert \mathbf{H}_{u,S}\mathbf{b}_u\rVert_2}\right\}.
\end{equation}


\begin{lemma}\label{lem:b}
The optimal solution of \eqref{eq:pb} is given by
\begin{equation}
\mathbf{b}^{\star}_u = 
\frac{\Re\{\mathbf{H}_{u,S}^H\mathbf{H}_{u,S}\}^{-1}\Re\{\mathbf{H}_{u,S}^H\mathbf{q}_{k(u)}\}}{\left\|\Re\{\mathbf{H}_{u,S}^H\mathbf{H}_{u,S}\}^{-1}\Re\{\mathbf{H}_{u,S}^H\mathbf{q}_{k(u)}\}\right\|_2}.
\end{equation}
\end{lemma}

\begin{proof}
For conciseness, denote
\begin{equation}\label{eq:Ga}
\left\{\begin{aligned}
\mathbf{G} &= \Re\{\mathbf{H}_{u,S}^H\mathbf{H}_{u,S}\},\\
\mathbf{a} &= \Re\{\mathbf{H}_{u,S}^H\mathbf{q}_{k(u)}\}.
\end{aligned}\right.
\end{equation}
Since \(\mathbf{b}_u\) is a real-valued vector, the objective function in \eqref{eq:pb} can be expressed with quadratic forms:
\begin{equation}
\Re\left\{\frac{\mathbf{q}_{k(u)}^H\mathbf{H}_{u,S}\mathbf{b}_u}{\lVert \mathbf{H}_{u,S}\mathbf{b}_u\rVert_2}\right\} \equiv \frac{\mathbf{a}^T\mathbf{b}_u}{\sqrt{\mathbf{b}^T_u\mathbf{G}\mathbf{b}_u}},
\end{equation}
where ${\bf a}$ and ${\bf G}$ are appropriately defined. It can be seen that this objective function depends only on the direction of \(\mathbf{b}_u\), not its magnitude, so the constraint \(\lVert \mathbf{b}_u\rVert_{2}=1\) can be removed. In what follows, the denominator of the objective function can be constrained to have a unit value as a constraint condition, and the problem is now converted to
\begin{equation}\label{eq:pb3}
\max_{\mathbf{b}_u\in\mathbb{R}^{K_u}}\mathbf{a^T\mathbf{b}_u}~~\mbox{s.t.}~~\mathbf{b}^T_u\mathbf{G}\mathbf{b}_u = 1.
\end{equation}
Its Lagrangian function is given by
\begin{equation}
{\cal L}(\mathbf{b}_u,\lambda) = \mathbf{a}^T\mathbf{b}_u - \lambda(\mathbf{b}^T_u\mathbf{G}\mathbf{b}_u-1).
\end{equation}
Setting its first derivative to zero, we have
\begin{equation}
\mathbf{b}_u \propto \mathbf{G}^{-1}\mathbf{a}.
\end{equation}
The solution can be obtained by normalization so that
\begin{equation}
\mathbf{b}^{\star}_u= \frac{\mathbf{G}^{-1}\mathbf{a}}{\rVert\mathbf{G}^{-1}\mathbf{a}\lVert} = 
\frac{\Re\{\mathbf{H}_{u,S}^H\mathbf{H}_{u,S}\}^{-1}\Re\{\mathbf{H}_{u,S}^H\mathbf{q}_{k(u)}\}}{\left\|\Re\{\mathbf{H}_{u,S}^H\mathbf{H}_{u,S}\}^{-1}\Re\{\mathbf{H}_{u,S}^H\mathbf{q}_{k(u)}\}\right\|_2},
\end{equation}
which completes the proof.
\end{proof}

Based on the closed-form solution of \(\mathbf{b}_u\), we can then deal with the port selection strategy. For a fixed \(\mathbf{q}_{k(u)}\), once the port set \(S\) is chosen, the optimal combining vector \(\mathbf{b}^{\star}_u\) can be obtained from \eqref{lem:b}. Thus, the ideal port set selection problem is to maximize the projection energy, i.e.,
\begin{equation}\label{eq:p2}
P_2:\left\{\begin{aligned}
\max_{S} &~\Re\left\{\mathbf{q}_{k(u)}^H\frac{\mathbf{H}_{u,S}\mathbf{b}_u}{\lVert \mathbf{H}_{u,S}\mathbf{b}_u\rVert_2}\right\}\\
\mbox{s.t.} &~\mathbf{H}_{u,S} = \mathbf{H}_{u}[:,S],\\
		&~S = \{1\leq s_1 < s_2 < \dots < s_{K_u} \leq N\}.
\end{aligned}\right.
\end{equation}
Solving $P_2$ would be prohibitively complex if all \(\binom{N}{K_{u}}\) subsets are visited exhaustively, especially when \(K_u\ll N\). Hence, a more effective method with high performance is required.

In fact, we can tackle the problem similar to the way we solve for the codeword index \(K_u\) in Section \ref{sec:codeword}. Notice that if \(\mathbf{b}_u\) is a complex-valued vector, then for each selected port set \(S\), \(\mathbf{b}_u\) will adaptively adjust itself based on \(\mathbf{H}_{u,S}\), ensuring that the value of the objective function equals the projection of \(\mathbf{q}_{k(u)}^H\) onto \(\operatorname{span}\{\mathbf{H}_{u,S}\}\). Specifically, denote the orthogonal projector of subspace \(\operatorname{span}\{\mathbf{H}_{u,S}\}\):
\begin{equation}
	\mathbf{P}_{S} = \mathbf{H}_{u,S}\big(\mathbf{H}_{u,S}^H\mathbf{H}_{u,S}\big)^{-1}\mathbf{H}_{u,S}^H.
\end{equation}
Then we consider the following problem:
\begin{equation}\label{eq:p3}
P_3:\left\{\begin{aligned}
\max_{S} &~\lVert\mathbf{P}_{S}\mathbf{q}_{k(u)}\rVert_2^2\\
\mbox{s.t.} &~\mathbf{H}_{u,S} = \mathbf{H}_{u}[:,S],\\
		&~S = \{1\leq s_1 < s_2 < \dots < s_{K_u} \leq N\},
\end{aligned}\right.
\end{equation}
which can give a suboptimal solution to \(P_2\) if \(\mathbf{b}_u\) is a complex-valued vector. However, \(\mathbf{b}_u\) is required to be a real-valued vector. In this case, the relevant geometry should be defined by the real inner product
\begin{equation}
\langle \mathbf{x},\mathbf{y}\rangle_{\mathbb{R}} = \Re\{\mathbf{x}^H\mathbf{y}\}.
\end{equation}
That is, the orthogonal projection of \( \mathbf{q}_{k(u)} \) onto \( \operatorname{span}\{\mathbf{H}_{u,S}\} \) equals the real-coefficient least-squares solution:
\begin{equation}
\mathbf{\Pi}_{S}^{(\mathbb{R})}\mathbf{q}_{k(u)}
=\mathbf{H}_{u,S}\big(\Re\{\mathbf{H}_{u,S}^H\mathbf{H}_{u,S}\}\big)^{-1}\Re\{\mathbf{H}_{u,S}^H\mathbf{q}_{k(u)}\}.
\end{equation}
Therefore, instead of maximizing \( \|\mathbf{P}_S\mathbf{q}_{k(u)}\|_2^2 \) with the complex projector, we formulate the following problem: 
\begin{equation}\label{eq:p4}
P_4:\left\{\begin{aligned}
\max_{S} &~\left\|\mathbf{\Pi}_{S}^{(\mathbb{R})}\mathbf{q}_{k(u)}\right\|_2^2\\
\mbox{s.t.} &~\mathbf{H}_{u,S} = \mathbf{H}_{u}[:,S],\\
		&~S = \{1\leq s_1 < s_2 < \dots < s_{K_u} \leq N\}.
\end{aligned}\right.
\end{equation}

A simple, and widely used approach to solve orthogonal projection problem like \(P_4\) is a greedy selection by orthogonal matching pursuit (OMP), which repeatedly selects the column most correlated with the current residual and then orthogonally projects \(\mathbf{q}_{k(u)}\) onto the span of the selected columns.

Let \(\mathbf{r}^{(t)}\) denote the residual after \(t\) selections, initialized by \(\mathbf{r}^{(0)}=\mathbf{q}_{k(u)}\). At iteration \(t+1\), for the current residual \(\mathbf{r}^{(t)}\), OMP considers adding a column \(\boldsymbol{h}_{u,i}\) into \(\mathbf{H}_{u,S}\). If the angle between \(\boldsymbol{h}_{u,i}\) and \(\mathbf{r}^{(t)}\) is smaller within the real inner product space, then \(\boldsymbol{h}_{u,i}\) can more effectively ``eliminate'' a portion of the residual. In particular, OMP chooses
\begin{equation}\label{eq:selection_rule}
i^{\star} \in \arg\max_{i\notin S^{(t)}} \frac{\left|\Re\{\boldsymbol{h}_{u,i}^H\mathbf{r}^{(t)}\}\right|}{\lVert\boldsymbol{h}_{u,i}\rVert_{2}},
\end{equation}
updates the chosen ports set \(S^{(t+1)}=S^{(t)}\cup\{i^{\star}\}\), and refreshes the residual by the orthogonal projection
\begin{equation}\label{eq:residual_update}
	\mathbf{r}^{(t+1)}=\mathbf{q}_{k(u)}-\mathbf{\Pi}^{(\mathbb{R})}_{S^{(t+1)}}\mathbf{q}_{k(u)}.
\end{equation}
The normalization in \eqref{eq:selection_rule} prevents bias towards large-norm columns. This strategy can iteratively select the vector that can minimize the residual the most and add it to the subspace. The detailed algorithm is presented in Algorithm \ref{alg:OMP}.

\begin{algorithm}[t]
\caption{OMP based port selection algorithm}\label{alg:OMP}
\begin{algorithmic}[1]
		\State \textbf{Input:} \(\mathbf{H}_u=\big[\boldsymbol{h}_{u,1},\dots,\boldsymbol{h}_{u,N}\big], \boldsymbol{q}_{k(u)}\), \(K_u\)
		\State \textbf{Pre-normalize:} For each \(i\), set \(\tilde{\boldsymbol{h}}_{u,i}=\boldsymbol{h}_{u,i}/\sqrt{\Re\{\boldsymbol{h}_{u,i}^H\boldsymbol{h}_{u,i}\}}\). Let \(\tilde{\mathbf{H}}_u=[\tilde{\boldsymbol{h}}_{u,1},\ldots,\tilde{\boldsymbol{h}}_{u,N}]\).
		\State Initialize \(S^{(0)}\leftarrow\varnothing\), \(\mathbf{r}^{(0)}\leftarrow \mathbf{q}_{K(u)}\).
		\For{\(t=1,\ldots,K_u\)}
		\State \(c_{i}\leftarrow \big|\Re\{\tilde{\mathbf{h}}_{u,i}^H\mathbf{r}^{(t-1)}\}\big|\) for all \(i\notin S^{(t-1)}\).
		\State Select \(i^{\star}\in\arg\max_{i\notin S^{(t-1)}} c_{i}\).
		\State Update \(S^{(t)}\leftarrow S^{(t-1)}\cup\{i^{\star}\}\) and \(\mathbf{H}_{u,S} = \mathbf{H}_{u}[:,S^{(t)}]\).
		\State Compute the real-coefficient projector \(\mathbf{\Pi}_{S^{(t)}}^{(\mathbb{R})}\mathbf{q}_{k(u)}
		=
		\mathbf{H}_{u,S^{(t)}}\big(\Re\{\mathbf{H}_{u,S^{(t)}}^H\mathbf{H}_{u,S^{(t)}}\}\big)^{-1}\Re\{\mathbf{H}_{u,S^{(t)}}^H\mathbf{q}_{k(u)}\}\).
		\State Update residual: \(\mathbf{r}^{(t)}\leftarrow \mathbf{q}_{k(u)}-\mathbf{\Pi}_{S^{(t)}}^{(\mathbb{R})}\mathbf{q}_{k(u)}\).
		\EndFor
		\State \textbf{Output:} Index set \(S^{(K_u)}\).
	\end{algorithmic}
\end{algorithm}

If the columns of \(\mathbf{H}_u\) are mutually orthogonal, OMP attains the optimum, with complexity approximately at \(O(MNK_u)\), which is much lower than the exhaustive method. 

\vspace{-2mm}
\subsection{No Combining Case with Fixed \(\mathbf{b}_u = \frac{1}{\sqrt{K_u}}\mathbf{1}\)}\label{subsec:ncb}
In some cases, the hardware cost of the FAS at UE needs to be lower and the passive attenuators at the ports for combining are not possible. Then we have \(\mathbf{b}_u = \frac{1}{\sqrt{K_u}}\mathbf{1}\) which is fixed. In this case, the port selection strategy should be adjusted. 

If the OMP-based algorithm is used, then since \(\mathbf{b}_u\) cannot adapt to each \(S\), it may introduce errors. Additionally, if \(\mathbf{b}_u\) is fixed, Problem \(P_2\) can be simplified to a more tractable form. We will study the selection strategy for \(S\) here. 

When \(\mathbf{b}_u = \frac{1}{\sqrt{K_u}}\mathbf{1}\), the objective function of Problem \(P_2\) can be simplified to
\begin{align}
\Re\left\{\mathbf{q}_{k(u)}^H\frac{\mathbf{H}_{u,S}\mathbf{b}}{\lVert \mathbf{H}_{u,S}\mathbf{b}\rVert_2}\right\} &= \Re\left\{\frac{\mathbf{q}_{k(u)}^H\sum_{n\in S}\boldsymbol{h}_{u,n}}{\sum_{p,q\in S}\boldsymbol{h}_{u,p}^H\boldsymbol{h}_{u,q}}\right\}\notag\\
&= \frac{\lvert\sum_{n\in S} g_n\rvert}{\sqrt{\sum_{p,q\in S} r_{pq}}},
\end{align}
where \(g_n = \mathbf{q}_{k(u)}^{H}\boldsymbol{h}_{u,n}\in\mathbb{C}\), and \(r_{pq} = \Re\{\boldsymbol{h}_{u,p}^{H}\boldsymbol{h}_{u,q}\}\in\mathbb{R}\).

Next, we consider using a greedy algorithm, i.e, each time, select a vector \(\boldsymbol{h}_{u,t}\) from the remaining candidate column vectors such that the current score, i.e., the objective function, is maximized. For conciseness, we define 
\begin{equation}
\left\{\begin{aligned}
G^{(t)} &\triangleq \sum_{n\in S^{(t)}} g_{n},\\
R^{(t)} &\triangleq \sum_{p,q\in S^{(t)}} r_{pq},
\end{aligned}\right.
\end{equation}
in which \(S^{(t)}\) denotes the index set after \(t\) selections. Assuming the \(t\)-th selection chooses the port \(\boldsymbol{h}_{u,l(t)}\), then the iterative forms of \(G^{(t)}\) and \(R^{(t)}\) can be expressed as
\begin{equation}
\left\{\begin{aligned}
G^{(t)}& =  G^{(t-1)} + g_{l(t)},\\
R^{(t)}& = R^{(t-1)} + 2\sum_{p\in S^{(t-1)}} r_{p{l(t)}} + r_{{l(t)}{l(t)}}.
\end{aligned}\right.
\end{equation}
Let \(\Delta^{(t)}\) be the greedy score after \(t\) selections, given by
\begin{equation}\label{eq:score}
\Delta^{(t)} = \frac{|G^{(t-1)} + g_{l(t)}|^{2}}{R^{(t-1)} + 2\sum_{p\in S^{(t-1)}} r_{p{l(t)}} + r_{{l(t)}{l(t)}}}.
\end{equation}
At each selection, pick \({l(t)}^{\star}=\arg\max_{{l(t)}\notin S^{(t-1)}}\Delta^{(t)}\). Also, since combining is not applied, there may be cases in which activating any additional port cannot improve performance. In other words, when \(t\geq 2\), \({l(t)}^{\star}\) needs to satisfy
\begin{equation}
\begin{aligned}
\frac{|G^{(t-1)} + g_{{l(t)}^{\star}}|^{2}}{R^{(t-1)} + 2\sum_{p\in S^{(t-1)}} r_{p{{l(t)}^{\star}}} + r_{{{l(t)}^{\star}}{{l(t)}^{\star}}}} > 
\frac{|G^{(t-1)}|^2}{R^{(t-1)}}.
\end{aligned}
\end{equation}
Otherwise, the selection process will terminate. The detailed algorithm is presented in Algorithm \ref{alg:NCC}.

\begin{algorithm}[t]
\caption{Greedy Port Selection for No Combining Case}\label{alg:NCC}
\begin{algorithmic}[1]
		\State \textbf{Input:} \(\mathbf{H}_u=\big[\boldsymbol{h}_{u,1},\dots,\boldsymbol{h}_{u,N}\big], \mathbf{q}_{k(u)}\).
		\State \textbf{Precompute} \(g_n = \mathbf{q}_{k(u)}^{H}\boldsymbol{h}_{u,n}\in\mathbb{C}\), and \(r_{pq} = \Re\{\boldsymbol{h}_{u,p}^{H}\boldsymbol{h}_{u,q}\}\) for all \(1 \leq n,p,q \leq N\).
		\State \(S^{(0)}\leftarrow\varnothing\), \(G^{(0)}\leftarrow 0\), \(R^{(0)}\leftarrow 0\), \(\Delta^{(0)}\leftarrow 0\).
		\While{\(|S|<K_{u}\)}
		\State \(t = t + 1\).
		\ForAll{\(l(t) \notin  S^{(t)}\)}
		\State \(G^{(t)} =  G^{(t-1)} + g_{l(t)}\).
		\State \(R^{(t)} = R^{(t-1)} + 2\sum_{p\in S^{(t-1)}} r_{p{l(t)}} + r_{{l(t)}{l(t)}}\).
		\State \(\Delta^{(t)} = \frac{G^{(t)}}{R^{(t)}}\).
		\EndFor
		\State Select \({l(t)}^{\star}=\arg\max_{{l(t)}\notin S^{(t-1)}}\Delta^{(t)}\), 
		\State \(\Delta^{(t)\star}=\max_{{l(t)}\notin S^{(t-1)}}\Delta^{(t)}\).
		\If{\(\Delta^{(t)\star} \leq \Delta^{(t-1)}\)} 
		\State \textbf{break} 
		\EndIf
		\State Update \(S^{(t)}\leftarrow S^{(t-1)}\cup\{l(t)^{\star}\}\), \(G^{(t)}\leftarrow G^{(t-1)} + g_{l(t)^{\star}}\), \(R^{(t)}\leftarrow R^{(t-1)} + 2\sum_{p\in S^{(t-1)}} r_{p{l(t)^{\star}}} + r_{{l(t)}^{\star}{l(t)^{\star}}}\), \(\Delta^{(t)}\leftarrow \Delta^{(t)\star}\).
		\EndWhile
		\State \textbf{Output:} Index set \(S^{(t)}\).
	\end{algorithmic}
\end{algorithm}

The computational complexity of this algorithm lies in the pre-computation, which requires \(O(MN+N^2)\) operations, while the complexity of the iterative process is \(O(NK_u)\). As a result, the overall complexity is \(O(N(M+N+K_u))\).

\vspace{-2mm}
\section{Codeword Collisions}\label{sec:collision}
To accommodate a large number of UEs without centralized joint codeword assignment, codeword collision is inevitable, i.e., the event that multiple UEs select the same codeword, and the BS cannot distinguish these two UEs' messages. A simple way to deal with this problem is to increase the number of antennas \(M\) to provide enough codewords. Assume \(U\) UEs select among \(M\) codewords uniformly, the probability that all \(U\) UEs pick distinct codewords is given by
\begin{equation}
P_{\text{unique}} = \frac{M(M-1)\cdots(M-U+1)}{M^{U}}.
\end{equation}
When \(M\gg U\), the probability can be asymptotically approximated as
\begin{equation}
P_{\text{unique}} \approx \exp\left(-\frac{U(U-1)}{2M}\right).
\end{equation}  

This exponential function shows that enlarging \(M\) when \(M\gg U\) quickly suppresses collisions. However, when \(U\) is very large, setting \(M\gg U\) could become a burden on the BS. It is therefore important to design alternative schemes. 

To tackle collision, we fine the ALOHA family of random-access protocols relevant. ALOHA was originally proposed for decentralized packet transmission in shared wireless channels \cite{ALOHA}. In ALOHA, users transmit packets in an uncoordinated manner. To improve efficiency, a reservation-based variant, known as reservation ALOHA (R-ALOHA), can be used \cite{RALOHA1,RALOHA2}. In R-ALOHA, each user first sends a short reservation message to the BS to request access, after which the BS schedules collision-free transmissions within the data phase. This scheme significantly reduces channel contention and improves throughput compared with conventional slotted ALOHA \cite{RALOHA3}.

For this reason, the reservation phase is used for codeword detection. In particular, each UE sends a very short signal, and the BS determines whether a collision occurs for each codeword by correlating the received signal with the corresponding codeword vector, i.e.,
\begin{equation}
r_k = \mathbf{q}_k^{H}\mathbf{y}_{\mathrm{reservation}},
\end{equation}
where \(\mathbf{y}_{\mathrm{reservation}}\) denotes the received signal vector collected within the reservation phase. If a single UE transmits using codeword \(\mathbf{q}_k\), the correlation magnitude \(\lvert r_k\rvert\) will be much higher than the noise level. On the contrary, if more than one UEs select \(\mathbf{q}_k\), then the received signal becomes a superposition of multiple channels, and the BS can detect the collision by measuring the energy \(E_k = \lvert r_k\rvert^2\) and comparing it with a predefined threshold derived from the noise power and single-user statistical powers. This projection-based method allows the BS to detect codeword collisions without explicit CSI. 

As a result, the BS can identify which codewords have been selected without performing complicated collision detection or multiuser separation. If multiple UEs choose the same codeword, a collision occurs. Then the BS broadcasts a short feedback message to inform the collided UEs. Depending on the system design, several strategies are possible: 
\begin{itemize}
\item[$(1)$] The BS simply instructs the UEs to defer transmission and randomly reattempt in a later slot.
\item[$(2)$] The BS chooses a UE to keep the codeword (based on the power or just chooses arbitrarily), and informs other UEs of the currently unoccupied codeword indices \(\mathbf{Q}_{\text{remain}}\) and let them autonomously select new ones. In this case, each remaining collided UE can follow the rule \eqref{eq:kus1} to select an alternative codeword. This strategy allows each UE to complete communication within the current time slot without postponement, while performance does not degrade significantly. However, if \(M\) is not sufficiently large, the remaining users still have a relatively high probability of selecting the same codeword, and with no further opportunity for reselection, a collision will inevitably occur. In such cases, the UEs involved in the collision are considered to have failed transmission. The BS will feed back a transmission failure indicator to them and ask for a retransmission in the next frame.
\item[$(3)$] The BS directly assigns specific alternative codewords from the unoccupied set to those UEs, thereby avoiding further collisions. However, since the BS has no CSI, it can only assign codewords to UEs randomly, which will inevitably lead to performance degradation.
\end{itemize}

Clearly, there exists a fundamental trade-off between signaling overhead and collision avoidance performance. The first strategy introduces almost no extra signal processing cost but suffers from possible retransmission delay. The second strategy improves slot efficiency at the cost of a small probability of re-collision, while the third strategy eliminates collisions completely but requires additional BS-side feedback to reassign codewords and suffers from the performance loss caused by arbitrary assignment of codewords. In practice, the optimal choice depends on system load and latency requirements. For instance, in massive access scenarios, the second strategy provides a near-optimal balance between complexity and performance. The paper will provide a comprehensive comparison of these three methods in the simulations section.

Recent studies have explored new directions to reduce codeword collisions in large-scale random access. One line of work adopts unsourced random access (URA), where all UEs share a common codebook and transmit short messages without explicit identifiers, while the BS jointly decodes the set of transmitted codewords from their superimposed observations \cite{Schedule1}. Recently in \cite{Schedule2} and \cite{Zhang-2025wcl}, URA was further combined with FAS and MIMO, which might offer an excellent solution to the codeword collision problem by providing additional spatial degrees of freedom for separating colliding UEs. Due to space limitations, only three basic schemes are introduced in this paper, while exploring more advanced like URA remains a promising direction for future research.

\vspace{-2mm}
\section{Simulation Results}\label{sec:sim}
Here, we present simulation results to evaluate the effectiveness of the proposed CPSC schemes. The simulations consider an uplink communication scenario as illustrated in Fig.~\ref{fig:model}. The main metric is the average data rate, defined as
\begin{equation}
\bar{R}_u = \mathbb{E}\{\log\left(1+{\rm SINR}_u\right)\},
\end{equation}
where the SINR is given by
\begin{equation}
{\rm SINR}_u =\frac{P_u\,\big|\mathbf{q}_{k(u)}^{H}\boldsymbol{h}_u\big|^{2}}
{\sum_{v=1\atop v\neq u}^{U} P_v\,\big|\mathbf{q}_{k(u)}^{H}\boldsymbol{h}_v\big|^{2} + \sigma^{2}}.
\end{equation}
Average rate reflects the system throughput under random user access and potential codeword collisions. The channels are modeled with Rice fading channels. The Rice factor is set as \(L = 0.1\). The UEs are randomly distributed on a $100{\rm m}\times 100{\rm m}$ plane, located $200{\rm m}$ away from the BS. The size of FAS is set as \(4\lambda\times 4\lambda\), in which the carrier wavelength \(\lambda\) is \(1{\rm cm}\), i.e., the carrier frequency is \(30\) GHz. The received signal-to-noise ratio (SNR) is fixed as \(10\) dB. The large-scale path loss is neglected. All simulation results are obtained with \(10^4\) Monte-Carlo independent trials. Unless otherwise specified, collisions are handled by a random deferral approach, where the colliding UEs retransmit in subsequent slots, and their instantaneous rates are counted as zero in the current slot. Also, the CPSC-FAMA scheme jointly optimizes the codeword index \(k(u)\), the port set \(S\), and the combining vector \(\mathbf{b}_u\) in sequence, while the fixed-antenna benchmark only optimizes \(k(u)\) and \(\mathbf{b}_u\). Additionally, for convenience, the number of activated ports \(K_u\) for all UEs is set to be equal, i.e., \(K_1=\dots=K_U=K\). In the CPSC-FAMA scheme, the codeword vector \(\mathbf{q}_{K(u)}\) is obtained through SVD, where the number of retained singular vectors \(t\) is fixed to \(K_u\) unless otherwise specified. The results provide insights into how port selection, combining design, and collision control jointly affect the system performance. 

Fig.~\ref{fig:RateN} illustrates the average achievable rate versus the number of available ports \(N\) for different transmission schemes. The number of UEs is fixed as \(U=16\). It can be observed from the figure that the proposed CPSC-FAMA scheme consistently outperforms the conventional fixed-antenna benchmark, owing to its ability to dynamically optimize the port-selection set to adjust the instantaneous channel vector to match the candidate codewords. In addition, the average rate increases with \(N\) as expected, as more ports provide additional spatial degrees of freedom and enhance the channel diversity. In contrast, the fixed-antenna scheme can only adjust the codeword and combining vector to match the instantaneous channel situation, leading to a noticeable performance gap. This improvement remains stable across different values of \(M\) and \(K\), confirming the robustness of the proposed CPSC-FAMA design in exploiting port-level spatial diversity. Moreover, the results indicate that increasing the number of BS antennas \(M\) expands the codebook size and thus improves the selection flexibility, while enlarging the number of active ports (or antennas) \(K\) provides more degrees of freedom for optimizing the combining vector. Both factors contribute to performance enhancement under the proposed and benchmark schemes.

\begin{figure}
\centering
\includegraphics[width=.95\linewidth]{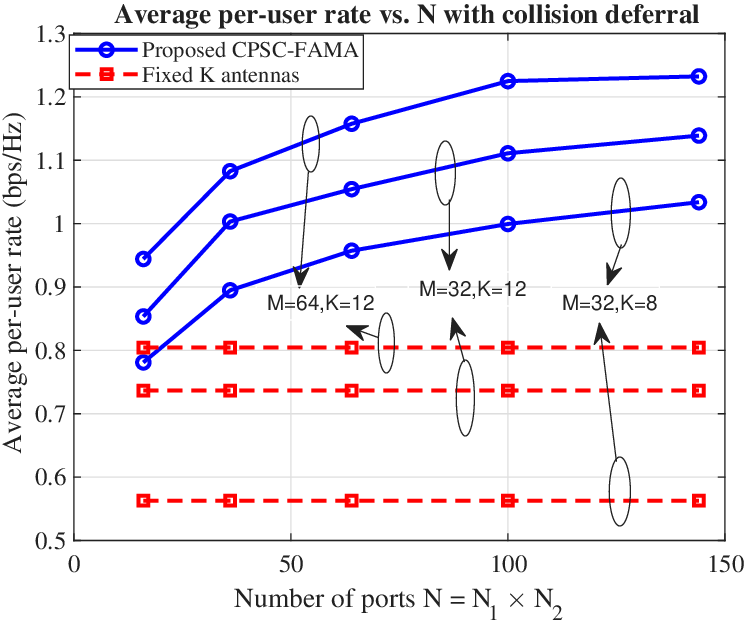}
\caption{Average rate per user vs. the number of ports \(N\), when the number of UEs is fixed as \(U=16\).}\label{fig:RateN}
\vspace{-2mm}
\end{figure}

In Figs.~\ref{fig:RateU} and~\ref{fig:sumRateU}, we illustrate the average per-user rate and the total sum rate versus the number of UEs \(U\), respectively. The number of activated ports is \(K=16\). The per-user average rate decreases with increasing \(U\) due to stronger inter-user interference and a higher probability of codeword collisions. The proposed CPSC-FAMA scheme consistently outperforms the fixed-antenna benchmark in both metrics, benefiting from the additional spatial diversity via adaptive port selection. Moreover, increasing the number of BS antennas \(M\) or the number of available ports \(N\) leads to further improvements for both schemes, as a larger \(M\) expands the available codebook and a larger \(N\) enhances the port-selection flexibility. When considering the sum-rate performance, it can be observed that increasing the number of UEs does not bring any overall gain for the fixed-antenna scheme, as the growing interference and collision probability quickly offset the benefit of multiuser diversity. In contrast, the proposed CPSC-FAMA scheme can still exploit the spatial flexibility offered by the FAS structure to achieve additional throughput gain when \(U\) is small to moderate. Although the sum rate of the proposed scheme also tends to saturate as \(U\) further increases, equipping either more BS antennas or available FAS ports effectively postpones this saturation. It can be expected that as long as \(N\) is sufficiently large and the number of BS antennas \(M\) exceeds the number of UEs, the proposed CPSC-FAMA approach is capable of supporting a massive number of UEs, demonstrating its strong potential for large-scale uplink multi-access systems.

\begin{figure}
\centering
\includegraphics[width=.95\linewidth]{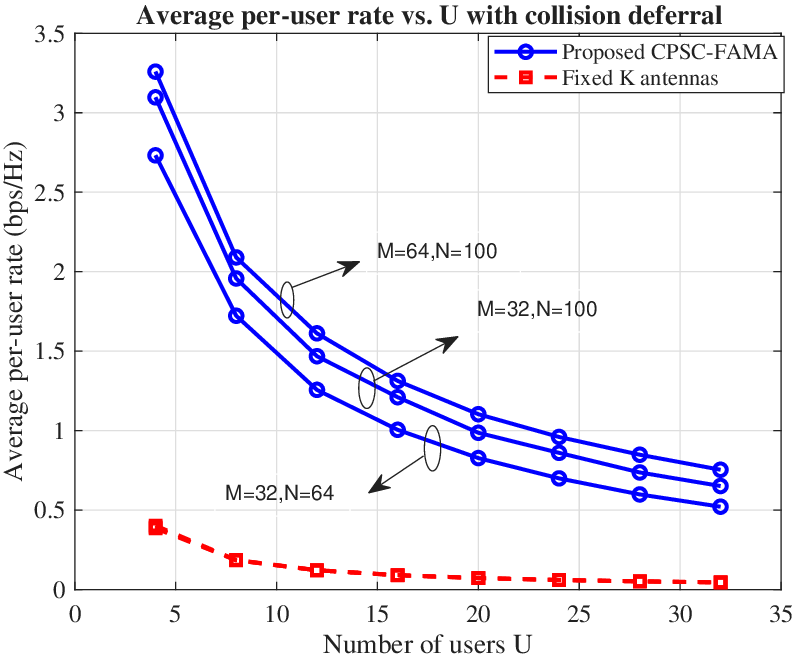}
\caption{Average rate per user vs. the number of UEs \(U\), when the number of activated ports is fixed as \(K=16\).}\label{fig:RateU}
\vspace{-2mm}
\end{figure}

\begin{figure}
\centering
\includegraphics[width=.95\linewidth]{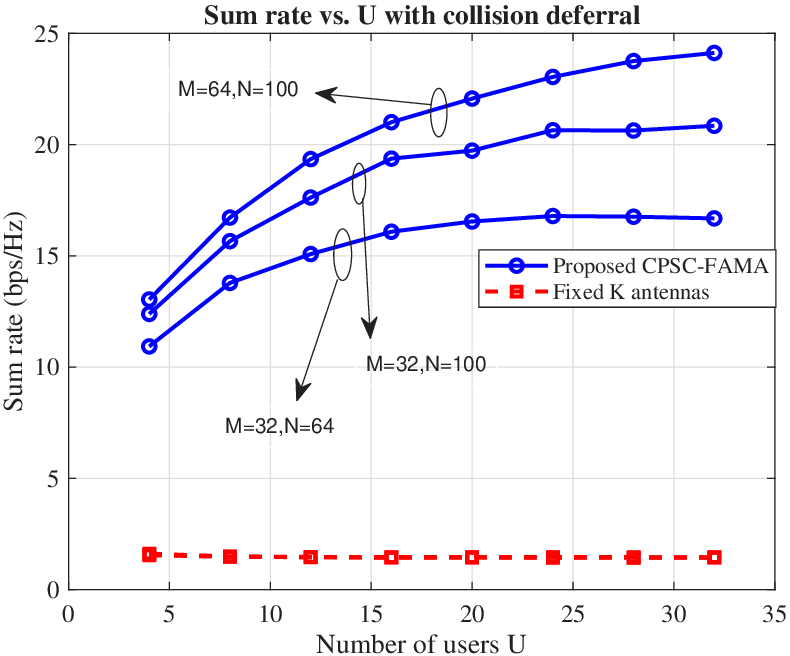}
\caption{Sum rate vs. the number of UEs \(U\), when the number of activated ports is fixed as \(K=16\).}\label{fig:sumRateU}
\vspace{-2mm}
\end{figure}

Fig.~\ref{fig:RateM} illustrates the average per-user rate versus the number of BS antennas \(M\) under different user and port configurations. The number of ports is fixed as \(N=100\). In addition to showing that the proposed CPSA-FAMA scheme consistently outperforms the fixed-antenna benchmark, Fig.~\ref{fig:RateM} also reveals that increasing \(M\) brings only a limited improvement in performance. The performance improvement brought by increasing \(M\) can be attributed to two main factors: (i) a larger number of available codewords, which provides each UE with more selection flexibility, and (ii) a lower probability of codeword collisions. As shown in the figure, when \(M\) is relatively small, increasing \(M\) yields a significant gain because the collision probability rapidly decreases with the expanded codebook. Nonetheless, as \(M\) continues to grow, the collision probability becomes very low, as illustrated in Fig.~\ref{fig:CRM}, and the remaining gain mainly comes from the availability of more codewords for selection. As a matter of fact, once the number of codewords far exceeds the number of UEs, each UE can already find a near-optimal codeword \(\mathbf{q}_{K(u)}\) that lies close to the column space of its channel matrix \(\mathbf{H}_{u}\), so further enlarging \(M\) brings little improvement. Therefore, it is sufficient for the BS to have a moderate number of antennas slightly larger than the number of UEs, rather than an excessively large array. This property stems from the CSI-free nature of the proposed CPSC-FAMA scheme, where the BS does not fully exploit the spatial degrees of freedom but instead delegates part of the adaptation burden to the UEs. Consequently, the proposed design greatly reduces the computational and signaling load at the BS, making it particularly suitable for large-scale user-access scenarios.

\begin{figure}
\centering
\includegraphics[width=.95\linewidth]{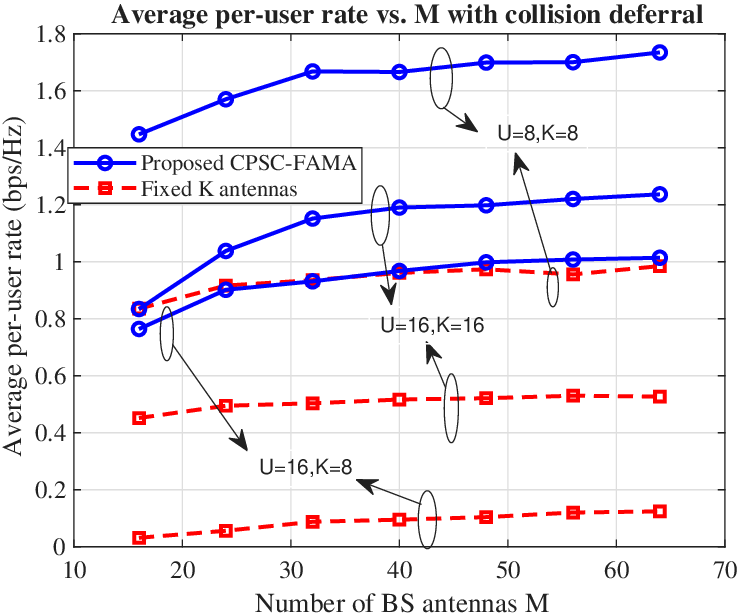}
\caption{Average rate per user vs. the number of BS antennas \(M\), when the number of ports is fixed as \(N=100\).}\label{fig:RateM}
\vspace{-2mm}
\end{figure}

Fig.~\ref{fig:CRM} shows the collision probability versus the number of BS antennas \(M\). As already discussed for Fig.~\ref{fig:RateM}, the collision rate decreases rapidly when \(M\) is small, because a larger number of BS antennas corresponds to a larger codebook and therefore a lower chance that multiple UEs select the same codeword. When \(M\) continues to increase, the collision probability quickly approaches a small value and remains almost unchanged. It can also be observed that varying the number of active ports \(K\) does not affect the collision probability, since collisions only depend on the number of users \(U\) and the total number of available codewords determined by \(M\).

\begin{figure}
\centering
\includegraphics[width=.95\linewidth]{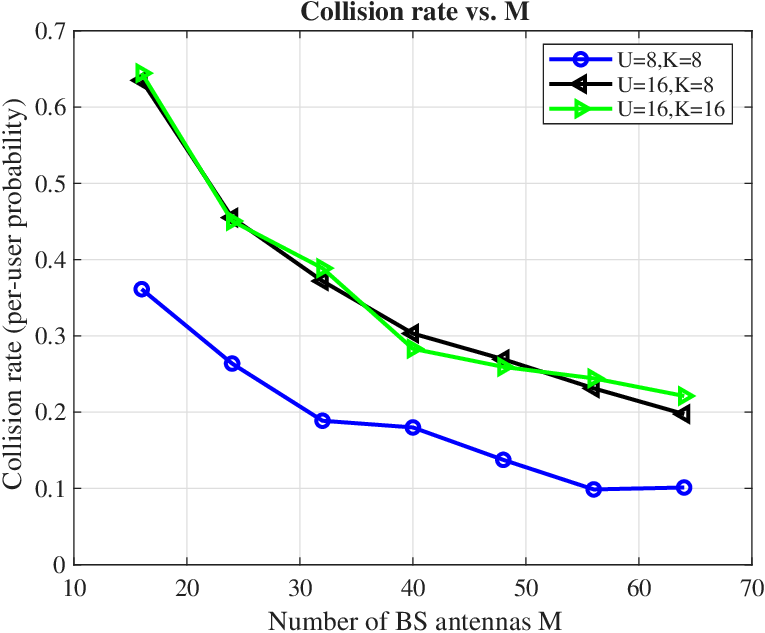}
\caption{Codeword collision rate vs. the number of BS antennas \(M\).}\label{fig:CRM}
\vspace{-2mm}
\end{figure}

Fig.~\ref{fig:Ratet} illustrates the average per-user rate versus the number of retained left singular vectors \(t\) in the SVD-based codeword optimization. The number of BS antennas and activated ports are fixed as \(M=64\) and \(K=8\), respectively. It can be observed that the impact of \(t\) on performance is generally minor. Even in the extreme case of \(t=1\), the performance degradation compared with a larger \(t\) is small, while the computational complexity can be significantly reduced. Moreover, increasing \(t\) does not always lead to monotonic improvement, the curves exhibit slight fluctuations as \(t\) grows, which may result from the inherent instability of low-rank channel estimation. Therefore, \(t\) should not be chosen to be too large. Instead, it should be selected according to practical trade-offs between performance and complexity.

\begin{figure}
\centering
\includegraphics[width=.95\linewidth]{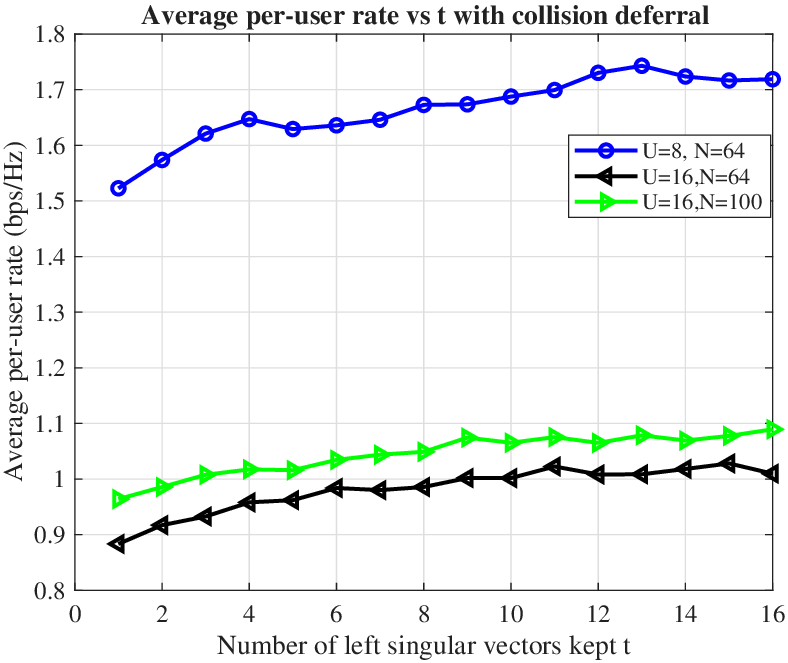}
\caption{Average rate per user vs. the number of left singular vectors \(t\), the number of BS antennas and activated ports are fixed as \(M=64\) and \(K=8\), respectively.}\label{fig:Ratet}
\vspace{-2mm}
\end{figure}

Fig.~\ref{fig:RateNbq} shows the average per-user rate versus the number of candidate ports \(N\) for the combining version of CPSC-FAMA in Section \ref{subsec:cb}, the non-combining version in Section \ref{subsec:ncb} of the proposed CPSA-FAMA scheme, and the arbitrarily assigning codewords strategy from Remark \ref{rem:nq}, along with the fixed-antenna benchmark. The number of BS antennas is fixed as \(M=64\). As expected, the non-combining scheme exhibits lower performance than the combining one, since it does not attenuate the different FAS ports. Nevertheless, it still achieves a substantial gain over the fixed-antenna case, highlighting the dominant contribution of port selection. Moreover, as \(N\) increases, the performance gap between the combining and non-combining schemes gradually narrows and eventually becomes negligible. This occurs because, when \(N\) is sufficiently large, the spatial degrees of freedom provided by the FAS structure already exceed the additional gain obtainable through combining. In such cases, the non-combining design can be adopted to save hardware resources without noticeable performance loss. As for the other scheme with arbitrarily assigning codewords at the BS, unlike the non-combining scheme, the performance gap between the optimized and randomly assigned codewords remains noticeable even when \(N\) is large. This indicates that random codeword allocation cannot fully exploit the spatial diversity provided by the FAS, while optimizing the codeword indices allows a better matching between the transmit and receive directions. That said, even with arbitrary codeword allocation, CPSA-FAMA still achieves a significantly higher rate than the fixed-antenna benchmark, due to the gain brought by the FAS-based port selection. Unlike the combining vector \(\mathbf{b}_u\), optimizing the codeword index does not require any additional hardware resources and can be easily implemented. Hence, when computational capability is sufficient, it is recommended that UEs perform optimized codeword selection to obtain a higher performance. These observations also confirm that the key advantage of the proposed CPSA-FAMA scheme lies in the FAS-based port selection process predominately.

\begin{figure}
\centering
\includegraphics[width=.95\linewidth]{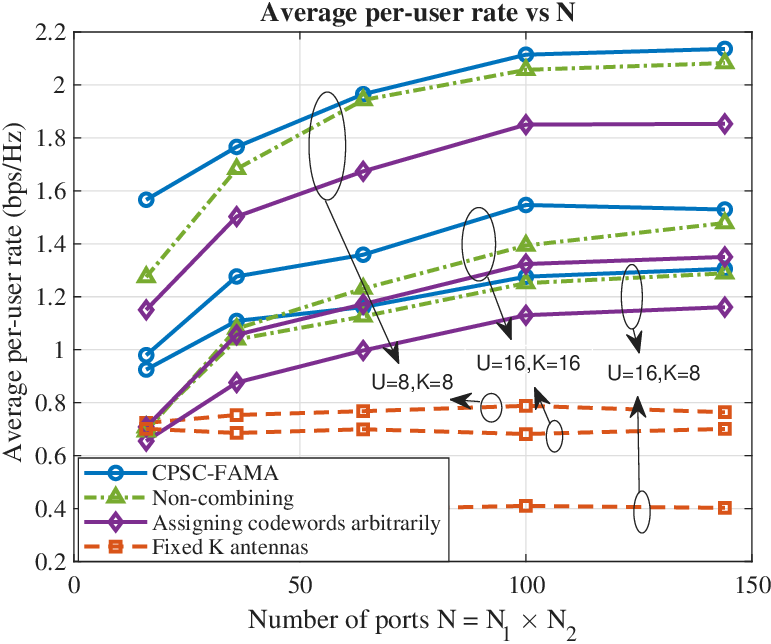}
\caption{Average rate per user vs. the number of ports \(N\) with four different schemes, including proposed CPSC-FAMA, non-combining with fixed \(\mathbf{b}_u\), BS assigning codewords arbitrarily, and conventional fixed antenna. The number of BS antennas is fixed as \(M=64\).}\label{fig:RateNbq}
\vspace{-2mm}
\end{figure}

Fig.~\ref{fig:Exh} compares the proposed CPSC-FAMA scheme with the exhaustive port-selection method and the fixed-antenna baseline. As shown, CPSC-FAMA with the OMP algorithm achieves performance very close to the exhaustive search. However, the exhaustive approach requires evaluating all possible port combinations, resulting in an explosion in complexity that makes it impractical for large-scale systems. Thus, the proposed low-complexity algorithm strikes an excellent trade-off between performance and computational efficiency.

\begin{figure}[t]
\centering
\includegraphics[width=.95\linewidth]{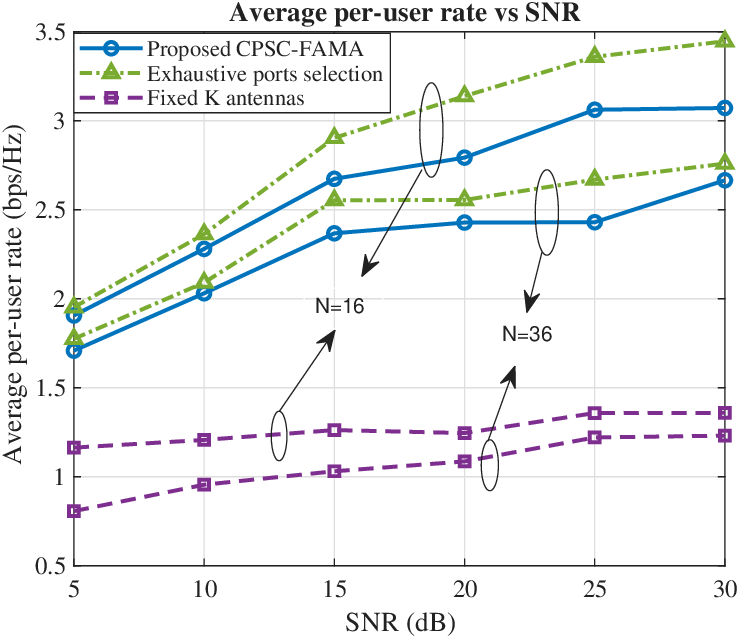}
\caption{Average rate per user vs. SNR with CPSC and exhaustive ports selection methods.}\label{fig:Exh}
\vspace{-2mm}
\end{figure}
\begin{figure}[t]
\centering
\includegraphics[width=.95\linewidth]{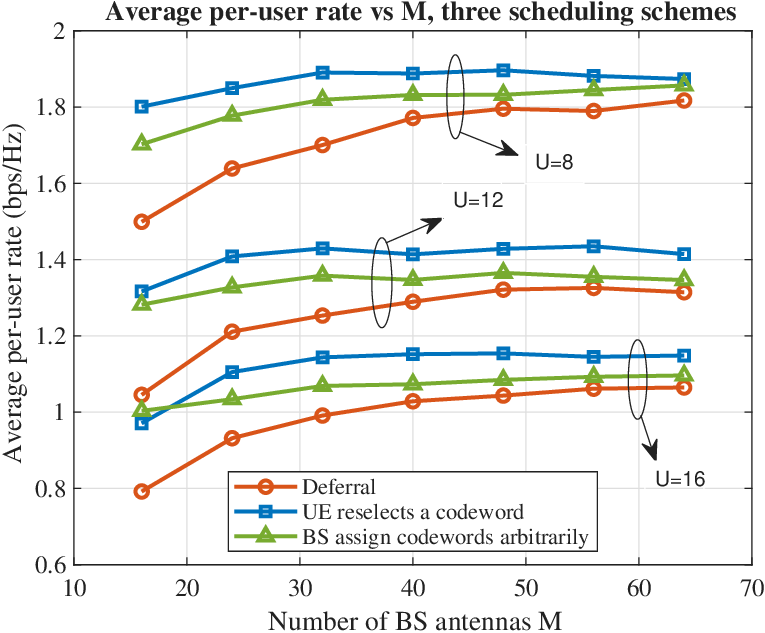}
\caption{Average rate per user vs. the number of BS antennas \(M\) with three different scheduling schemes.}\label{fig:Scheduling}
\vspace{-2mm}
\end{figure}

In Fig.~\ref{fig:Scheduling}, we illustrate the average per-user rate versus the number of BS antennas \(M\) under three different collision-handling schemes: random deferral, BS-side codeword reassignment, and UE-side codeword reselection. Amongst them, the deferral scheme has the worst performance as the colliding UEs waste transmission opportunities in the time slot. The BS-side reassignment performs better but not the best, because the BS has no CSI and can only reallocate codewords randomly, which has been discussed in Fig.~\ref{fig:RateNbq}. Thus, although this approach completely eliminates further collisions, it still causes a noticeable performance loss. The UE-side reselection scheme achieves the highest rate, as each UE can independently choose a more suitable codeword based on its local channel condition. Nevertheless, the performance gap among the three schemes narrows as \(M\) increases because a larger number of BS antennas provides more available codewords and substantially reduces the collision probability. These results indicate that while the reselection scheme provides the best performance, BS-assisted scheduling already offers a favorable trade-off between efficiency and complexity, which could be adopted according to the practical communication scenario.

\section{Conclusion}\label{sec:conclu}
This paper presented a novel CPSC-FAMA framework for multiuser uplink communication in the absence of CSI at the BS. The proposed scheme allows each UE to independently select codewords and activate ports based on local CSI, while the BS separates user signals through codebook-guided projection without CSI. Three lightweight scheduling strategies were designed to flexibly handle codeword collisions. Simulation results demonstrated that the proposed framework achieves a much higher average rate compared to fixed-antenna systems, while significantly reducing BS-side complexity. The findings highlight the strong potential of FAS for supporting massive and low-complexity access in future 6G networks.

\bibliographystyle{IEEEtran}

\end{document}